\documentclass{entcs} \usepackage{entcsmacro}
\pdfoutput=1 
\usepackage{graphicx}
\usepackage{appendix}
\usepackage{bussproofs} %
\usepackage{tikz}

\def\lastname{Ramos, de Queiroz, de Oliveira}
\begin{document}
\begin{frontmatter}
  \title{On the Groupoid Model of Computational Paths} \author{Arthur F. Ramos\thanksref{myemail}}
  \address{Centro de Informática\\ Universidade Federal de Pernambuco\\
    Recife, Brazil}
\author{Ruy J. G. B. de Queiroz\thanksref{coemail}}
  \address{Centro de Informática\\Universidade Federal de Pernambuco\\
  Recife, Brazil}
\author{Anjolina G. de Oliveira\thanksref{cocoemail}}
  \address{Centro de Informática\\Universidade Federal de Pernambuco\\
  Recife, Brazil}
 \ \thanks[myemail]{Email:
    \href{afr@cin.ufpe.br} {\texttt{\normalshape
        afr@cin.ufpe.br}}} \thanks[coemail]{Email:
    \href{ruy@cin.ufpe.br} {\texttt{\normalshape
        ruy@cin.ufpe.br}}}
    \thanks[cocoemail]{Email:\href{ago@cin.ufpe.br} {\texttt{\normalshape
        ago@cin.ufpe.br}}}

\begin{abstract}
	
	The main objective of this work is to study mathematical properties of computational paths. Originally proposed by de Queiroz \& Gabbay (1994) as `sequences or rewrites', computational paths are taken to be terms of the identity type of Martin L\"of's Intensional Type Theory, since these paths can be seen as the grounds on which the propositional equality between two computational objects stand. From this perspective, this work aims to show  that one of the properties of the identity type is present on computational paths. We are referring to the fact that that the identity type induces a groupoid structure, as proposed by Hofmann \& Streicher (1994). Using categorical semantics, we show that computational paths induce a groupoid structure. We also show that computational paths are capable of inducing higher categorical structures.

\end{abstract}
\end{frontmatter}
\section{Introduction}\label{intro}

There seems to be little doubt that the identity type is one of the most intriguing concepts of  Martin-L\"of's Type Theory. This claim is supported by recent groundbreaking discoveries. In 2005, Vladimir Voevodsky \cite{Vlad1} discovered the Univalent Models, resulting in a new area of research known as Homotopy Type Theory \cite{Steve1}. This theory is based on the fact that a term of some identity type, for example $p: Id_{A}(a,b)$, has a clear homotopical interpretation. The interpretation is that the witness $p$ can be seen as a homotopical path between the points $a$ and $b$ within a topological space $A$. This simple interpretation has made clear the connection between Type Theory and Homotopy Theory, generating groundbreaking results, as one can see in \cite{hott,Steve1}. It is important to emphasize that one important fact of the homotopic interpretation is that the homotopic paths exist only in the semantic sense. In other words, there is no formal entity in type theory that represents these paths. They are not present in the syntax of Type Theory.

In this work, we are interested in a concept known as computational path, originally proposed by \cite{Ruy4} and further explored by \cite{Ruy5}. A computational path is an entity that establishes the equality between two terms of the same type. It differs from the homotopical path, since it is not only a semantic interpretation. It is a formal entity of the equality theory, as proposed by \cite{Ruy5}. In fact, in our work entitled "On the Identity Type as the Type of Computational Paths" \cite{Art1}, we claim that computational paths should be considered as a formal entity of Type Theory and that the identity type is just the type of this new entity. That way, the computational path would exist not only in a semantic sense, but it would be also a formal entity of the syntax of Type Theory.

Since the connection between computational paths and identity types have been studied in the work that we have just mentioned, this connection will not be the objective here. Rather, we are interested in important mathematical properties that arise naturally from the concept of computational paths. Specifically, we want to show that computational paths naturally induce a structure known as a groupoid. In fact, that the identity type induces a groupoid structure is a well-known fact, established by \cite{hofmann1} in 1994. Nevertheless, this work differs substantially from \cite{hofmann1}. The main difference is the fact that in \cite{hofmann1}, the induced groupoid is constructed by the construction of several terms of the identity type, using the elimination rule (i.e., using the constructor $J$). In fact, we will not use the identity type as originally formulated or any elimination rule. We will work directly with the concept of computational path and concepts related to it. We will show that, using these concepts, we can induce a category that has the properties of a groupoid.

\section{Computational Paths} \label{path}

Since computational path is a generic term, it is important to emphasize the fact that we are using the term computational path in the sense defined by \cite{Ruy5}. A computational path is based on the idea that it is possible to formally define when two computational objects $a,b : A$ are equal. These two objects are equal if one can reach $b$ from $a$ applying a sequence of axioms or rules. This sequence of operations forms a path. Since it is between two computational objects, it is said that this path is a computational one. Also, an application of an axiom or a rule transforms (or rewrite) an term in another. For that reason, a computational path is also known as a sequence of rewrites. Nevertheless, before we define formally a computational path, we can take a look at one famous equality theory, the $\lambda\beta\eta-equality$ \cite{lambda}:

\begin{definition}
The \emph{$\lambda\beta\eta$-equality} is composed by the following axioms:

\begin{enumerate}
\item[$(\alpha)$] $\lambda x.M = \lambda y.[y/x]M$ \quad if $y \notin FV(M)$;
\item[$(\beta)$] $(\lambda x.M)N = [N/x]M$;
\item[$(\rho)$] $M = M$;
\item[$(\eta)$] $(\lambda x.Mx) = M$ \quad $(x \notin FV(M))$.
\end{enumerate}

And the following rules of inference:

\bigskip
\noindent
\begin{bprooftree}
\AxiomC{$M = M'$ }
\LeftLabel{$(\mu)$ \quad}
\UnaryInfC{$NM = NM'$}
\end{bprooftree}
\begin{bprooftree}
\AxiomC{$M = N$}
\AxiomC{$N = P$}
\LeftLabel{$(\tau)$}
\BinaryInfC{$M = P$}
\end{bprooftree}

\bigskip
\noindent
\begin{bprooftree}
\AxiomC{$M = M'$ }
\LeftLabel{$(\nu)$ \quad}
\UnaryInfC{$MN = M'N$}
\end{bprooftree}
\begin{bprooftree}
\AxiomC{$M = N$}
\LeftLabel{$(\sigma)$}
\UnaryInfC{$N = M$}
\end{bprooftree}

\bigskip
\noindent
\begin{bprooftree}
\AxiomC{$M = M'$ }
\LeftLabel{$(\xi)$ \quad}
\UnaryInfC{$\lambda x.M= \lambda x.M'$}
\end{bprooftree}




\end{definition}





\begin{definition}(\cite{lambda})
$P$ is $\beta$-equal or $\beta$-convertible to $Q$  (notation $P=_\beta Q$)
iff $Q$ is obtained from $P$ by a finite (perhaps empty)  series of $\beta$-contractions
and reversed $\beta$-contractions  and changes of bound variables.  That is,
$P=_\beta Q$ iff \textbf{there exist} $P_0, \ldots, P_n$ ($n\geq 0$)  such that
$P_0\equiv P$,  $P_n\equiv Q$,
$(\forall i\leq n-1) (P_i\triangleright_{1\beta}P_{i+1}  \mbox{ or }P_{i+1}\triangleright_{1\beta}P_i  \mbox{ or } P_i\equiv_\alpha P_{i+1}).$
\end{definition}
\noindent (NB: equality with an \textbf{existential} force, which will show in the proof rules for the identity type.)

The same happens with $\lambda\beta\eta$-equality:\\
\begin{definition}($\lambda\beta\eta$-equality \cite{lambda})
The equality-relation determined by the theory $\lambda\beta\eta$ is called
$=_{\beta\eta}$; that is, we define
$$M=_{\beta\eta}N\quad\Leftrightarrow\quad\lambda\beta\eta\vdash M=N.$$
\end{definition}

\begin{example}
Take the term $M\equiv(\lambda x.(\lambda y.yx)(\lambda w.zw))v$. Then, it is $\beta\eta$-equal to $N\equiv zv$ because of the sequence:\\
$(\lambda x.(\lambda y.yx)(\lambda w.zw))v, \quad  (\lambda x.(\lambda y.yx)z)v, \quad   (\lambda y.yv)z , \quad zv$\\
which starts from $M$ and ends with $N$, and each member of the sequence is obtained via 1-step $\beta$- or $\eta$-contraction of a previous term in the sequence. To take this sequence into a {\em path\/}, one has to apply transitivity twice, as we do in the example below.
\end{example}

\begin{example}\label{examplepath}
The term $M\equiv(\lambda x.(\lambda y.yx)(\lambda w.zw))v$ is $\beta\eta$-equal to $N\equiv zv$ because of the sequence:\\
$(\lambda x.(\lambda y.yx)(\lambda w.zw))v, \quad  (\lambda x.(\lambda y.yx)z)v, \quad   (\lambda y.yv)z , \quad zv$\\
Now, taking this sequence into a path leads us to the following:\\
The first is equal to the second based on the grounds:\\
$\eta((\lambda x.(\lambda y.yx)(\lambda w.zw))v,(\lambda x.(\lambda y.yx)z)v)$\\
The second is equal to the third based on the grounds:\\
$\beta((\lambda x.(\lambda y.yx)z)v,(\lambda y.yv)z)$\\
Now, the first is equal to the third based on the grounds:\\
$\tau(\eta((\lambda x.(\lambda y.yx)(\lambda w.zw))v,(\lambda x.(\lambda y.yx)z)v),\beta((\lambda x.(\lambda y.yx)z)v,(\lambda y.yv)z))$\\
Now, the third is equal to the fourth one based on the grounds:\\
$\beta((\lambda y.yv)z,zv)$\\
Thus, the first one is equal to the fourth one based on the grounds:\\
$\tau(\tau(\eta((\lambda x.(\lambda y.yx)(\lambda w.zw))v,(\lambda x.(\lambda y.yx)z)v),\beta((\lambda x.(\lambda y.yx)z)v,(\lambda y.yv)z)),\beta((\lambda y.yv)z,zv)))$.
\end{example}


The aforementioned theory establishes the equality between two $\lambda$-terms. Since we are working with computational objects as terms of a type, we need to translate the $\lambda\beta\eta$-equality to a suitable equality theory based on Martin L\"of's type theory. We obtain:
	
\begin{definition}
The equality theory of Martin L\"of's type theory has the following basic proof rules for the $\Pi$-type:

\bigskip

\noindent
\begin{bprooftree}
\hskip -0.3pt
\alwaysNoLine
\AxiomC{$N : A$}
\AxiomC{$[x : A]$}
\UnaryInfC{$M : B$}
\alwaysSingleLine
\LeftLabel{$(\beta$) \quad}
\BinaryInfC{$(\lambda x.M)N = M[N/x] : B[N/x]$}
\end{bprooftree}
\begin{bprooftree}
\hskip 11pt
\alwaysNoLine
\AxiomC{$[x : A]$}
\UnaryInfC{$M = M' : B$}
\alwaysSingleLine
\LeftLabel{$(\xi)$ \quad}
\UnaryInfC{$\lambda x.M = \lambda x.M' : (\Pi x : A)B$}
\end{bprooftree}

\bigskip

\noindent
\begin{bprooftree}
\hskip -0.5pt
\AxiomC{$M : A$}
\LeftLabel{$(\rho)$ \quad}
\UnaryInfC{$M = M : A$}
\end{bprooftree}
\begin{bprooftree}
\hskip 100pt
\AxiomC{$M = M' : A$}
\AxiomC{$N : (\Pi x : A)B$}
\LeftLabel{$(\mu)$ \quad}
\BinaryInfC{$NM = NM' : B[M/x]$}
\end{bprooftree}

\bigskip

\noindent
\begin{bprooftree}
\hskip -0.5pt
\AxiomC{$M = N : A$}
\LeftLabel{$(\sigma) \quad$}
\UnaryInfC{$N = M : A$}
\end{bprooftree}
\begin{bprooftree}
\hskip 105pt
\AxiomC{$N : A$}
\AxiomC{$M = M' : (\Pi x : A)B$}
\LeftLabel{$(\nu)$ \quad}
\BinaryInfC{$MN = M'N : B[N/x]$}
\end{bprooftree}

\bigskip

\noindent
\begin{bprooftree}
\hskip -0.5pt
\AxiomC{$M = N : A$}
\AxiomC{$N = P : A$}
\LeftLabel{$(\tau)$ \quad}
\BinaryInfC{$M = P : A$}
\end{bprooftree}

\bigskip

\noindent
\begin{bprooftree}
\hskip -0.5pt
\AxiomC{$M: (\Pi x : A)B$}
\LeftLabel{$(\eta)$ \quad}
\RightLabel {$(x \notin FV(M))$}
\UnaryInfC{$(\lambda x.Mx) = M: (\Pi x : A)B$}
\end{bprooftree}

\bigskip

\end{definition}

We are finally able to formally define computational paths:

\begin{definition}
Let $a$ and $b$ be elements of a type $A$. Then, a \emph{computational path} $s$ from $a$ to $b$ is a composition of rewrites (each rewrite is an application of the inference rules of the equality theory of type theory or is a change of bound variables). We denote that by $a =_{s} b$.
\end{definition}

As we have seen in example \ref{examplepath}, composition of rewrites are applications of the rule $\tau$. Since change of bound variables is possible, consider that each term is considered up to $\alpha$-equivalence.

\section{A Term Rewriting System for Paths}

As we have just showed, a computational path establishes when two terms of the same type are equal. From the theory of computational paths, an interesting case arises. Suppose we have a path $s$ that establishes that $a =_{s} b : A$ and a path $t$ that establishes that $a =_{t} b : A$. Consider that $s$ and $t$ are formed by distinct compositions of rewrites. Is it possible to conclude that there are cases that $s$ and $t$ should be considered equivalent? The answer is \emph{yes}. Consider the following examples:

\begin{example}
	\noindent \normalfont Consider the path  $a =_{t} b : A$. By the symmetric property, we obtain $b =_{\sigma(t)} a : A$. What if we apply the property again on the path $\sigma(t)$? We would obtain a path  $a =_{\sigma(\sigma(t))} b : A$. Since we applied symmetry twice in succession, we obtained a path that is equivalent to the initial path $t$. For that reason, we conclude the act of applying symmetry twice in succession is a redundancy. We say that the path $\sigma(\sigma(t))$ can be reduced to the path $t$.
\end{example}

\begin{example}
	\noindent \normalfont  Consider the reflexive path $a =_\rho a : A$. Since it is a trivial path, the symmetric path $a =_{\sigma(\rho)} a : A$ is equivalent to the initial one. For that reason, the application of the symmetry on the reflexive path is a redundancy. The path $\sigma(\rho)$ can be reduced to the path $\rho$.
	
\end{example}
\begin{example}
	\noindent \normalfont Consider the path $a =_{t} b : A$ and the inverse path $b =_{\sigma(t)} a : A$. We can apply the transitive property in these paths, obtaining $a =_{\tau(t,\sigma(t))} a : A$. Since the paths are inversions of each other, the transitive path is equivalent to the trivial path $\rho$. Therefore, this transitive application is a redundancy. The path $\tau(t,\sigma(t))$ can be reduced to the trivial path $\rho$.
\end{example}

As one could see in the aforementioned examples, different paths should be considered equal if one is just a redundant form of the other. The examples that we have just seen are straightforward and simple cases. Since the equality theory has a total of 7 axioms, the possibility of combinations that could generate redundancies are high. Fortunately, all possible redundancies were thoroughly mapped by De Oliveira (1995)\cite{Anjo1}. In this work, a system that establishes all redundancies and creates rules that solve them was proposed. This system, known as $LND_{EQ}-TRS$, maps a total of 39 rules that solve redundancies. Since we have a special interest on the groupoid model, we are not interested in all redundancy rules, but in a very specific subset of these rules (all have been taken from \cite{Anjo1,Ruy1}):

\begin{itemize}
\item Rules involving $\sigma$ and $\rho$

\bigskip

\begin{prooftree}
\AxiomC{$x =_{\rho} x : A$}
\RightLabel{\quad $\rhd_{sr}$ \quad $x =_{\rho} x : A$}
\UnaryInfC{$x =_{\sigma(\rho)} x : A$}
\end{prooftree}

\begin{prooftree}
\AxiomC{$x =_{r} y : A$}
\UnaryInfC{$y =_{\sigma(r)} x : A$}
\RightLabel{\quad $\rhd_{ss}$ \quad $x =_{r} y : A$}
\UnaryInfC{$x =_{\sigma(\sigma(r))} y : A$}
\end{prooftree}

\bigskip

\item Rules involving $\tau$

\bigskip
\begin{prooftree}
\AxiomC{$x =_{r} y : A$}
\AxiomC{$y =_{\sigma(r)} x : A$}
\RightLabel{\quad  $\rhd_{tr}$ \quad $x =_{\rho} x : A$}
\BinaryInfC{$x =_{\tau(r,\sigma(r))} x : A$}
\end{prooftree}

\begin{prooftree}
\AxiomC{$y =_{\sigma(r)} x : A$}
\AxiomC{$x =_{r} y : A$}
\RightLabel{\quad $\rhd_{tsr}$ \quad $y =_{\rho} y : A$}
\BinaryInfC{$y =_{\tau(\sigma(r),r)} y : A$}
\end{prooftree}

\begin{prooftree}
\AxiomC{$x =_{r} y : A$}
\AxiomC{$y =_{\rho} y : A$}
\RightLabel{\quad $\rhd_{trr}$ \quad $x =_{r} y : A$}
\BinaryInfC{$x =_{\tau(r,\rho)} y : A$}
\end{prooftree}

\begin{prooftree}
\AxiomC{$x =_{\rho} x : A$}
\AxiomC{$x =_{r} y : A$}
\RightLabel{\quad $\rhd_{tlr}$ \quad $x =_{r} y : A$}
\BinaryInfC{$x =_{\tau(\rho,r)} y : A$}
\end{prooftree}

\bigskip

\item Rule involving $\tau$ and $\tau$

\begin{prooftree}
\hskip - 155pt
\AxiomC{$x =_{t} y : A$}
\AxiomC{$y =_{r} w : A$}
\BinaryInfC{$x =_{\tau(t,r)} w : A$}
\AxiomC{$w =_{s} z : A$}
\BinaryInfC{$x =_{\tau(\tau(t,r),s)} z : A$}
\end{prooftree}

\begin{prooftree}
\hskip 4cm
\AxiomC{$x =_{t} y : A$}
\AxiomC{$y=_{r} w : A$}
\AxiomC{$w=_{s} z : A$}
\BinaryInfC{$y =_{\tau(r,s)} z : A$}
\LeftLabel{$\rhd_{tt}$}
\BinaryInfC{$x =_{\tau(t,\tau(r,s))} z : A$}
\end{prooftree}

\bigskip

\end{itemize}

\begin{definition}
\normalfont An $rw$-rule is any of the rules defined in $LND_{EQ}-TRS$.
\end{definition}

\begin{definition}
Let $s$ and $t$ be computational paths. We say that $s \rhd_{1rw} t$ (read as: $s$ $rw$-contracts to $t$) iff we can obtain $t$ from $s$ by an application of only one $rw$-rule. If $s$ can be reduced to $t$ by finite number of $rw$-contractions, then we say that $s \rhd_{rw} t$ (read as $s$ $rw$-reduces to $t$).

\end{definition}

\begin{definition}
\normalfont  Let $s$ and $t$ be computational paths. We say that $s =_{rw} t$ (read as: $s$ is $rw$-equal to $t$) iff $t$ can be obtained from $s$ by a finite (perhaps empty) series of $rw$-contractions and reversed $rw$-contractions. In other words, $s =_{rw} t$ iff there exists a sequence $R_{0},....,R_{n}$, with $n \geq 0$, such that

\centering $(\forall i \leq n - 1) (R_{i}\rhd_{1rw} R_{i+1}$ or $R_{i+1} \rhd_{1rw} R_{i})$

\centering  $R_{0} \equiv s$, \quad $R_{n} \equiv t$
\end{definition}

\begin{proposition}\label{proposition3.7}  is transitive, symmetric and reflexive.
\end{proposition}

\begin{proof}
Comes directly from the fact that $rw$-equality is the transitive, reflexive and symmetric closure of $rw$.
\end{proof}

Before talking about higher $LND_{EQ}-TRS$ systems, we'd like to mention that  $LND_{EQ}-TRS$ is terminating and confluent. The proof of this affirmation can be found in \cite{Anjo1,Ruy2,Ruy3,RuyAnjolinaLivro}.

\subsection{$LND_{EQ}-TRS_{2}$}

Until now, this section has concluded that there exist redundancies which are resolved by a system called $LND_{EQ}-TRS$. This system establishes rules that reduces these redundancies. Moreover, we concluded that these redundancies are just redundant uses of the equality axioms showed in \emph{section 2}. In fact, since these axioms just defines an equality theory for type theory, one can specify and say that these are redundancies of the equality of type theory. As we mentioned, the $LND_{EQ}-TRS$ has a total of 39 rules \cite{Anjo1}. Since the $rw$-equality is based on the rules of  $LND_{EQ}-TRS$,  one can just imagine the high number of redundancies that $rw$-equality could cause. In fact, a thoroughly study of all the redundancies caused by these rules could generate an entire new work. Fortunately, we are only interested in the redundancies caused by the fact that $rw$-equality is transitive, reflexive and symmetric with the addition of only one specific $rw_{2}$-rule. Let's say that we have a system, called $LND_{EQ}-TRS_{2}$, that resolves all the redundancies caused by $rw$-equality (the same way that $LND_{EQ}-TRS$ resolves all the redundancies caused by equality). Since we know that $rw$-equality is transitive, symmetric and reflexive, it should have the same redundancies that the equality had involving only these properties. Since $rw$-equality is just a sequence of $rw$-rules (also similar to equality, since equality is just a computational path, i.e., a sequence of identifiers), then we could put a name on these sequences. For example, if $s$ and $t$ are $rw$-equal because there exists a sequence $\theta: R_{0},....,R_{n}$ that justifies the $rw$-equality, then we can write that $s =_{rw_{\theta}} t$. Thus, we can rewrite, using $rw$-equality, all the rules that originated the rules involving $\tau$, $\sigma$ and $\rho$. For example, we have:

\bigskip

\begin{prooftree}
\hskip - 155pt
\AxiomC{$x =_{rw_{t}} y : A$}
\AxiomC{$y =_{rw_{r}} w : A$}
\BinaryInfC{$x =_{rw_{\tau(t,r)}} w : A$}
\AxiomC{$w =_{rw_{s}} z : A$}
\BinaryInfC{$x =_{rw_{\tau(\tau(t,r),s)}} z : A$}
\end{prooftree}

\begin{prooftree}
\hskip 4cm
\AxiomC{$x =_{rw_{t}} y : A$}
\AxiomC{$y=_{rw_{r}} w : A$}
\AxiomC{$w=_{rw_{s}} z : A$}
\BinaryInfC{$y =_{rw_{\tau(r,s)}} z : A$}
\LeftLabel{$\rhd_{tt_{2}}$}
\BinaryInfC{$x =_{rw_{\tau(t,\tau(r,s))}} z : A$}
\end{prooftree}

\bigskip

Therefore, we obtain the rule $tt_{2}$, that resolves one of the redundancies caused by the transitivity of $rw$-equality (the $2$ in $tt_{2}$ indicates that it is a rule that resolves a redundancy of $rw$-equality). In fact, using the same reasoning, we can obtain, for $rw$-equality, all the redundancies that we have showed in page $5$. In other words, we have $tr_{2}$, $tsr_{2}$, $trr_{2}$, $tlr_{2}$, $sr_{2}$, $ss_{2}$ and $tt_{2}$. Since we have now rules of $LND_{EQ}-TRS_{2}$, we can use all the concepts that we have just defined for $LND_{EQ}-TRS$. The only difference is that instead of having $rw$-rules and $rw$-equality, we have $rw_{2}$-rules and $rw_{2}$-equality.

There is a important rule specific to this system. It stems from the fact that transitivity of reducible paths can be reduced in different ways, but generating the same result. For example, consider the simple case of $\tau(s,t)$ and consider that it is possible to reduce $s$ to $s'$ and $t$ to $t'$. There is two possible $rw$-sequences that reduces this case: The first one is $\theta: \tau(s,t) \rhd_{1rw} \tau(s',t) \rhd_{1rw} \tau(s',t')$ and the second $\theta': \tau(s,t) \rhd_{1rw} \tau(s,t') \rhd_{1rw} \tau(s',t')$. Both $rw$-sequences obtained the same result in similar ways, the only difference being the choices that have been made at each step. Since the variables, when considered individually,  followed the same reductions, these $rw$-sequences should be considered redundant relative to each other and, for that reason, there should be $rw_{2}$-rule that establishes this reduction. This rule is called \emph{choice independence} and is denoted by $cd_{2}$. In fact, independent of the quantity of transitivities and variables, if the sole difference between the $rw$-sequences are the choices that were made in each step, then this rule will establish the $rw_{2}$-equality between the sequences.

\begin{proposition}
$rw_{2}$-equality is transitive, symmetric and reflexive.
\end{proposition}

\begin{proof}
Analogous to Proposition \ref{proposition3.7}.
\end{proof}

\section{The Groupoid Induced by Computational Paths}

The objective of this section is to show that computational paths, together with the reduction rules discussed in the last section, are capable of inducing structures known as groupoids. To do that, we will use a categorical interpretation.

Before we conclude that computational paths induces categories with groupoid properties, we need to make clear the difference between a strict and a weak category. As one will see, the word weak will appear many times. This will be the case because some of the categorical equalities will not hold ``on the nose", so to say. They will hold up to $rw$-equality or up to higher levels of $rw$-equalities. This is similar to the groupoid model of the identity type proposed by \cite{hofmann1}. In \cite{hofmann1}, the equalities do not hold ``on the nose", they hold up to propositional equality (or up to homotopy if one uses the homotopic interpretation). To indicate that these equalities hold only up to some property, we say that the induced structure is a weak categorical structure.

For each type $A$, computational paths induces a weak categorical structure called $A_{rw}$. In $A_{rw}$, the objects are terms $a$ of the type $A$ and morphisms between terms $a : A$ and $b : A$ are arrows $s: a \rightarrow b$ such that $s$ is a computational path between the terms, i.e., $a =_{s} b : A$. One can easily check that $A_{rw}$ is a category. To do that, just define the composition of morphisms as the transitivity of two paths and the reflexive path as the identity arrow. The $rw$-equality of the associativity comes directly from the $tt$ rule and the $rw$-equality of the identity laws comes from the $trr$ and $tlr$ rules. Of course, as discussed before, this only forms a weak structure, since the equalities only hold up to $rw$-equality.  The most interesting fact about $A_{rw}$ is the following proposition:

\begin{proposition}
The induced structure $A_{rw}$  has a weak groupoidal structure.
\end{proposition}

\begin{proof}
A groupoid is just a category in which every arrow is an isomorphism. Since we are working in a weak sense, the isomorphism equalities need only to hold up to $rw$-equality. To show that, for every arrow $s: a \rightarrow b$, we need to show a $t: b \rightarrow a$ such that $t \circ s =_{rw} 1_{a}$ and $s \circ t =_{rw} 1_{b}$. To do that, recall that every computational path has an inverse path $\sigma(s)$. Put $t = \sigma(s)$. Thus, we have that $s \circ t = s \circ \sigma(s) = \tau(\sigma(s),s) \rhd_{tsr} \rho_{b}$. Since $\rho_{b} = 1_{b}$, we conclude that  $s \circ t =_{rw} 1_{b}$. We also have that $t \circ s = \sigma(s) \circ s = \tau(s,\sigma(s)) \rhd_{tr} \rho_{a}$. Therefore,  $t \circ s = 1_{a}$. We conclude that every arrow is a weak isomorphism and thus, $A_{rw}$ is a weak groupoidal structure.
\end{proof}

With that, we conclude that computational paths have a groupoid model.

\subsection{Higher structures}

We have just showed that computational paths induce a weak groupoidal structure known as $A_{rw}$. We also know that the arrows (or morphisms) of $A_{rw}$ are computational paths between two terms of a type $A$. As we saw in the previous section, sometimes a computational path can be reduced to another by rules that we called $rw$-rules. That way, if we have terms $a, b : A$ and paths between these terms, we can define a new structure. This new structure, called $A_{2rw}(a,b)$, has, as objects, computational paths between $a$ and $b$ and the set of morphisms $Hom_{A_{2rw}(a,b)}(s,t)$ between paths $a =_{s} b$ and $a =_{t} b$ corresponds to the set of sequences that prove $s =_{rw} t$. Since $rw$-equality is transitive, reflexive and symmetric, $A_{2rw}$ is a weak categorical structure which the equalities hold up to $rw_{2}$-equality. The proof of this fact is analogous to the \emph{proposition 4.1}. The sole difference is the fact that since the morphisms are $rw$-equalities, instead of computational paths, all the equalities will hold up to $rw_{2}$-equality. To see this, take the example of the associativity. Looking at the $LND_{EQ}-TRS_{2}$ system, we have that $\tau(\tau(\theta,\sigma),\phi) \rhd_{tt_{2}} \tau(\theta,\tau(\sigma,\phi))$ ($\theta$, $\sigma$ and $\phi$ represent $rw$-equalities between paths from $a$ to $b$). Therefore,   $\tau(\tau(\theta,\sigma),\phi) =_{rw_{2}} \tau(\theta,\tau(\sigma,\phi))$. The associative law holds up to $rw_{2}$-equality. As one can easily check, the identity law will also hold up to $rw_{2}$-equality. Therefore, $A_{2rw}(a,b)$ has a weak categorical structure. Analogous to \emph{proposition 4.4}, the groupoid law will also hold up to $rw_{2}$-equality.

Instead of considering $A_{rw}$ and $A_{2rw}(a,b)$ as separated structures, we can think of a unique structure with 2-levels. The first level is $A_{rw}$ and the second one is the structure $A_{2rw}(a,b)$ between each pair of objects $a, b \in A_{rw}$. We call this structure $2-A_{rw}$. The morphisms of the first level are called $1$-morphisms and the ones of the second level are called $2$-morphisms (also known as $2$-arrows or $2$-cells). Since it has multiple levels, it is considered a higher structure. We want to prove that this structure is a categorical structure known as weak $2$-category. The main problem is the fact that in a weak $2$-category, the last level (i.e., the second level) needs to hold up in a strict sense. This is not the case for $2-A_{2rw}$, since each $A_{2rw}(a,b)$ only holds up to $rw_{2}$-equality. Nevertheless, there still a way to induce this weak 2-category. Since $rw$-equality is an equivalence relation (because it is transitive, symmetric and reflexive), we can consider a special $A_{2rw}(a,b)$, where the arrows are the arrows of $A_{2rw}(a,b)$ modulo $rw_{2}$-equality. That way, since the equalities hold up to $rw_{2}$-equality in $A_{2rw}(a,b)$, they will hold in a strict sense when we consider the equivalence classes of $rw_{2}$-equality. We call this structure $[A_{2rw}(a,b)]$. In this structure, consider the composition of arrows defined as: $[\theta]_{rw_{2}} \circ[\phi]_{rw_{2}} = [\theta \circ \phi]_{rw_{2}}$. Now,  we can think of the structure $[2-A_{rw}]$. This structure is similar to $2-A_{rw}$. The difference is that the categories of the second level are $[A_{2rw}(a,b)]$ instead of $A_{2rw}(a,b)$. We can now prove that $[2-A_{rw}]$ is a weak $2$-category:

\begin{proposition}

Computational paths induce a weak $2$-category called $[2-A_{rw}]$.

\end{proposition}

\begin{proof} First of all, let's draw a diagram that represents $[2-A_{rw}]$:

\bigskip

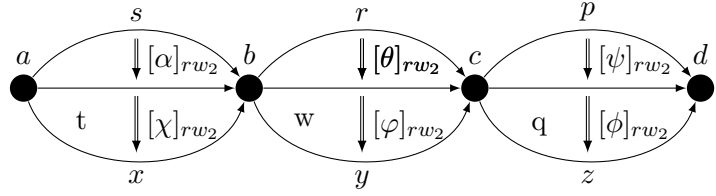
\begin{center}

\begin{tikzpicture}[>=latex,
    every node/.style={auto},
    arrowstyle/.style={double,->,shorten <=3pt,shorten >=3pt},
    mydot/.style={circle,fill}]

    \coordinate[mydot,label=above:$a$](a)  at (0,0);
    \coordinate[mydot,label=above:$b$](b) at (3,0);
    \coordinate[mydot,label=above:$c$](c) at (6,0);
    \coordinate[mydot,label=above:$d$](d) at (9,0);
    \draw[->]   (a) to[bend left=50]  coordinate (s) node[]{$s$}          (b);
    \draw[->]   (a) to coordinate (t) node[label={[label distance=-1cm]43:t}] {}    (b);
    \draw[->]   (a) to[bend right=70]  coordinate (x) node[,swap]{$x$}          (b);
    \draw[->]   (b) to[bend left=50]  coordinate (r) node[]{$r$}          (c);
    \draw[->]   (b) to coordinate (w) node[label={[label distance=-1cm]40:w}] {}    (c);
    \draw[->]   (b) to[bend right=70]  coordinate (y) node[,swap]{$y$}          (c);
    \draw[->]   (c) to[bend left=50]  coordinate (p) node[]{$p$}          (d);
    \draw[->]   (c) to coordinate (q) node[label={[label distance=-1cm]50:q}] {}    (d);
       \draw[->]   (c) to[bend right=70]  coordinate (z) node[,swap]{$z$}          (d);

    \draw[arrowstyle] (s) to node[right]{[$\alpha]_{rw_{2}}$} (t);
     \draw[arrowstyle] (t) to node[right]{[$\chi]_{rw_{2}}$} (x);
    \draw[arrowstyle] (r) to node[right]{[$\theta]_{rw_{2}}$} (w);
 \draw[arrowstyle] (w) to node[right]{[$\varphi]_{rw_{2}}$} (y);
\draw[arrowstyle] (r) to node[right]{[$\theta]_{rw_{2}}$} (w);
   \draw[arrowstyle] (p) to node[right]{[$\psi]_{rw_{2}}$} (q);
\draw[arrowstyle] (q) to node[right]{[$\phi]_{rw_{2}}$} (z);
\end{tikzpicture}	

\end{center}

In this diagram we represent $1$-arrows and $2$-arrows between these $1$-arrows. The fact that $2$-arrows are equivalence classes is represented by the brackets.

Given $[\alpha]_{rw_{2}}: [s = \alpha_{1}, ..., \alpha_{n} = t]$ and $[\theta]_{rw_{2}}: [r = \theta_{1}, ...,\theta_{m} = w]$, then we define the horizontal composition $([\theta]_{rw_{2}} \circ_{h} [\alpha]_{rw_{2}})$ as the sequence $[\tau(s = \alpha_{1}, r = \theta_{1}), ..., \tau(\alpha_{n},\theta_{1}) ..., \tau(\alpha_{n} = t,\theta_{m} = w)]_{rw_{2}}$.

We also need to verify the associative and identity law for $\circ_{h}$. Since we are working with a weak $2$-category, these laws should hold up to natural isomorphism \cite{leinster1}. To verify these laws, the idea is that every $2$-morphism of $[A_{rw_{2}}]$ is an isomorphism. The proof of this fact is analogous to the one of \emph{Proposition 5.1}, but using $rw_{2}$-rules instead of $rw$-rules. Since a natural transformation is a natural isomorphism iff every component is an isomorphism (as one can check in \cite{Steve2}), we conclude that finding isomorphisms for the associative and identity laws is just a matter of finding the correct morphisms.

For the associative law, we need to check that there is a natural isomorphism $assoc$ between $(([\psi]_{rw_{2}} \circ_{h} [\theta]_{rw_{2}}) \circ_{h} [\alpha]_{rw_{2}})$ and  $([\psi]_{rw_{2}} \circ_{h} ([\theta]_{rw_{2}} \circ_{h} [\alpha]_{rw_{2}}))$. To do this, by the definition of horizontal composition, a component of $(([\psi]_{rw_{2}} \circ_{h} [\theta]_{rw_{2}})$ is a term of the form $\tau(\alpha_{x}, \tau(\theta_{y}, \psi_{z}))$, with $x, y$, and $z$ being suitable natural numbers that respect the order of horizontal composition. Analogously, the same component of  $([\psi]_{rw_{2}} \circ_{h} ([\theta]_{rw_{2}} \circ_{h} [\alpha]_{rw_{2}}))$ is just a suitable term  $\tau(\tau(\alpha_{x},\theta_{y}), \psi_{z})$. The isomorphism between these component is clearly established by the inverse $tt$ rule, i.e.,  $\tau(\alpha_{x}, \tau(\theta_{y}, \psi_{z})) =_{rw_{\sigma(tt)}} \tau(\tau(\alpha_{x},\theta_{y}), \psi_{z})$

The identity laws use the same idea. We need to check that $([\alpha]_{rw_{2}} \circ_{h} [\rho_{\rho_{a}}]_{rw_{2}}) =  [\alpha]_{rw_{2}}$. To do that, we need to take components $(\rho_{\rho_{a}}, \alpha_{y})$ and $\alpha_{y}$ and establish their isomorphism  $r^{*}_{s}$: $(\rho_{\rho_{a}}, \alpha_{y}) =_{rw_{tlr}} \alpha_{y}$.

The other natural isomorphism  $l^{*}_{s}$, i.e., the isomorphism between  $( [\rho_{\rho_{b}}]_{rw_{2}} \circ_{h} [\alpha]_{rw_{2}})$ and $[\alpha]_{rw_{2}}$ can be established in an analogous way, just using the rule $trr$ instead of $tlr$. Just for purpose of clarification, $\rho_{\rho_{a}}$ comes from the reflexive property of $rw$-equality. Since $\rho_{a}$ is the identity path, using the reflexivity we establish that $\rho_{a} =_{rw} \rho_{a}$, generating $\rho_{\rho_{a}}$.

 With the associative and identity isomorphisms established, we now need to check the \emph{interchange law} \cite{leinster1}. We need to check that:
\begin{center}
$([\varphi]_{rw_{2}} \circ [\theta]_{rw_{2}}) \circ_{h} ([\chi]_{rw_{2}} \circ [\alpha]_{rw_{2}}) = ([\varphi]_{rw_{2}} \circ_{h} [\chi]_{rw_{2}}) \circ ([\theta]_{rw_{2}} \circ_{h} [\alpha]_{rw_{2}})$
\end{center}

From  $(([\varphi]_{rw_{2}} \circ [\theta]_{rw_{2}}) \circ_{h} ([\chi]_{rw_{2}} \circ [\alpha]_{rw_{2}}))$, we have:

\begin{center}
 $(([\varphi]_{rw_{2}} \circ [\theta]_{rw_{2}}) \circ_{h} ([\chi]_{rw_{2}} \circ [\alpha]_{rw_{2}})) =$

 $[\tau(\theta, \varphi)]_{rw_{2}} \circ_{h} [\tau(\alpha, \chi)]_{rw_{2}}=$

$[\theta_{1}, ..., \theta_{n} = \varphi_{1}, ..., \varphi_{n'}]_{rw_{2}} \circ_{h} [(\alpha_{1}, ..., \alpha_{m} = \chi_{1}, ... \chi_{m'}]_{rw_{2}} = $

$[\tau(\alpha_{1}, \theta_{1}),..., \tau(\alpha_{m} = \chi_{1}, \theta_{1}), ..., \tau(\chi_{n},\theta_{1}), ..., \tau(\chi_{n},\theta_{m'} = \varphi_{1}),..., \tau(\chi_{n},\varphi_{n'})]_{rw_{2}}$
\end{center}

From   $(([\varphi]_{rw_{2}} \circ_{h} [\chi]_{rw_{2}}) \circ ([\theta]_{rw_{2}} \circ_{h} [\alpha]_{rw_{2}})))$:

\begin{center}
 $(([\varphi]_{rw_{2}} \circ_{h} [\chi]_{rw_{2}}) \circ ([\theta]_{rw_{2}} \circ_{h} [\alpha]_{rw_{2}}))(r \circ s) =$

 $([\tau(\chi_{1}, \varphi_{1}),...,\tau(\chi_{n},\varphi_{1}),...,\tau(\chi_{n},\varphi_{n'})]_{rw_{2}} \circ [\tau(\alpha_{1},\theta_{1}),...,\tau(\alpha_{m},\theta_{1}),...,\tau(\alpha_{m},\theta_{m'})]_{rw_{2}} = $

$[\tau(\alpha_{1},\theta_{1}),...,\tau(\alpha_{m},\theta_{1}),...,\tau(\alpha_{m},\theta_{m'}),...,\tau(\chi_{1}, \varphi_{1}),...,\tau(\chi_{n},\varphi_{1}),...,\tau(\chi_{n},\varphi_{n'})]_{rw_{2}}$

\end{center}

If one looks closely, one can notice that this is a suitable to apply $cd_{2}$. Individually, every variable that appears in the sequence of transitivities follows the same expansion in both cases. The only difference is how the choices have been made. Therefore, the $rw_{2}$-equality is established by $cd_{2}$. Since we are working with equivalence classes, this equality holds strictly.

About the coherence laws, one can easily check that Mac Lane's pentagon and triangle diagrams are commutative by simple and straightforward applications of natural isomorphisms $assoc$, $r_{s}^{*}$ and $l_{s}^{*}$. One can check appendix A for a detailed proof of this part.
\end{proof}

We can also conclude that $[2-A_{rw}]$ has a weak $2$-groupoid structure. That is the case because we already know (from \emph{proposition 4.4}) that the groupoid laws are satisfied by $1$-morphisms up to the isomorphism of the next level, i.e., up to $rw$-equality and the $2$-morphisms, as we have just seen, are isomorphisms (that hold in a strict way, since the second level is using classes of equivalence).

Our objective in a future work is to define higher levels of $rw$-equalities and, from that, we will try to obtain even higher groupoid structures. Eventually, our goal will be the construction of a weak $\omega$-groupoid. We believe that it is possible to achieve these results, since it was proved by \cite{lumsdaine1,Benno} that the identity type induces such structure. Given the connection between computational paths and terms of identity type, we should be able to prove that computational paths also induces a weak $\omega$-groupoid.

\section{Conclusion}

The main purpose of this work was the study of mathematical properties of an entity known as computational paths. Inspired by our previous work, which has established the connection between computational paths and the identity type \cite{Art1}, our main motivation was the belief that some of the properties of the identity type of Martin L\"of's Intensional Type Theory should be present on computational paths. The property that we have investigated is the fact that the identity type induces a groupoid structure. Our idea has been the fact that computational paths should also be capable of naturally inducing such structure.

To achieve our results, we have defined a computational path as a sequence of equality identifiers between two terms of a type, where each identifier is either an axiom of the equality theory of Type Theory or a change of bound variables. Then, we have showed, using examples, that the usage of the equality axioms can cause redundancies within a computational path. We mentioned that there already exists a system called $LND_{EQ}-TRS$ that maps and resolves these redundancies. We have gone further, proposing the existence of higher $LND_{EQ}-TRS_{n}$ systems, using the idea that a $LND_{EQ}-TRS_{n}$ resolves the redundancies of a $LND_{EQ}-TRS_{n-1}$. Using  $LND_{EQ}-TRS$, we have proved that a computational path is capable of inducing a weak groupoid structure. Using the higher rewriting systems, we have showed that it is possible to induce a structure known as weak $2$-groupoid. Based on the fact that there exist infinite $LND_{EQ}-TRS_{n}$ systems, we opened the way, in a future work, for a possible proof establishing that computational paths induces a weak $w$-groupoid.

\bibliographystyle{plain}
\bibliography{ref}

\newpage
\begin{appendices}

\section{Coherence Laws}

In this section, we prove that Mac Lane's pentagon holds for $[2 - A_{rw}]$. In other words, we should check the following fact:

 \cite{leinster1}:
 
 Given the following diagram of $1-morphisms$:
 
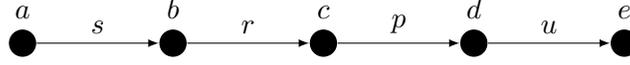
\begin{center}
 	\begin{tikzpicture}[>=latex,
 	every node/.style={auto},
 	arrowstyle/.style={double,->,shorten <=3pt,shorten >=3pt},
 	mydot/.style={circle,fill}]
 	
 	\coordinate[mydot,label=above:$a$](a)  at (0,0);
 	\coordinate[mydot,label=above:$b$](b) at (2,0);
 	\coordinate[mydot,label=above:$c$](c) at (4,0);
 	\coordinate[mydot,label=above:$d$](d) at (6,0);
 	\coordinate[mydot,label=above:$e$](e) at (8,0);
 	
 	\draw[->]   (a) to coordinate (s) node[] {$s$}    (b);
 	\draw[->]   (b) to coordinate (r) node[] {$r$}    (c);
 	\draw[->]   (c) to coordinate (p) node[] {$p$}    (d);
 	\draw[->]   (d) to coordinate (u) node[] {$u$}    (e);

 	\end{tikzpicture}
 \end{center}
 
 The following diagrams should commute:
 
 
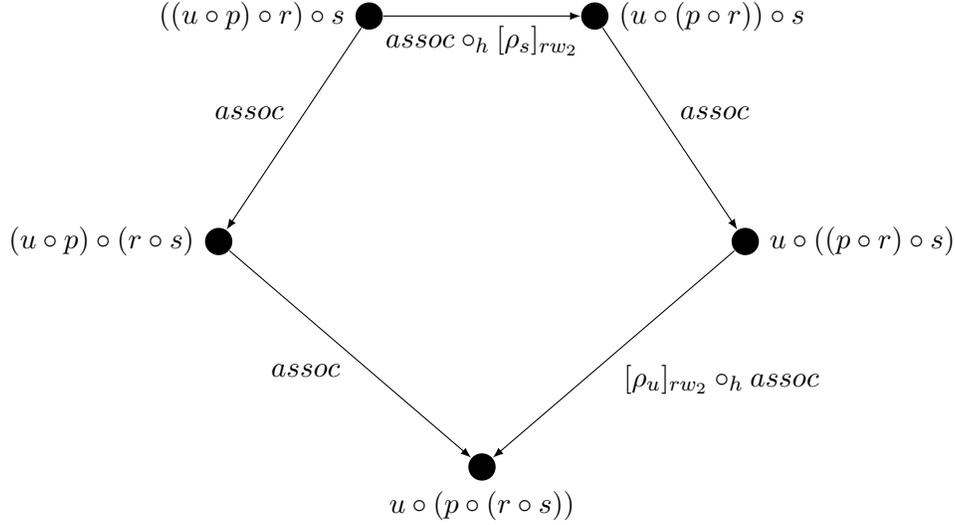
\begin{center}
 	\begin{tikzpicture}[>=latex,
 	every node/.style={auto},
 	arrowstyle/.style={double,->,shorten <=3pt,shorten >=3pt},
 	mydot/.style={circle,fill}]
 	
 	\coordinate[mydot,label=left:$((u\circ p) \circ r) \circ s$](a)  at (2,0);
 	\coordinate[mydot,label=left:$(u \circ p) \circ (r \circ s)$](b) at (0,-3);
 	\coordinate[mydot,label=right:$(u \circ(p \circ r)) \circ s$](c) at (5,0);
 	\coordinate[mydot,label=right:$u \circ((p \circ r) \circ s) $](d) at (7,-3);
 	\coordinate[mydot,label=below:$u \circ (p \circ (r \circ s))$](e) at (3.5,-6);
 	
 	\draw[->]   (a) to coordinate (s) node[,swap] {$assoc$}    (b);
 	\draw[->]   (a) to coordinate (t) node[,swap] {$assoc \circ_{h} [\rho_{s}]_{rw_{2}}$}    (c);
 	\draw[->]   (c) to coordinate (x) node[] {$assoc$}    (d);
 	\draw[->]   (b) to coordinate (y) node[,swap] {$assoc$}    (e);
 	\draw[->]   (d) to coordinate (z) node[] {$[\rho_{u}]_{rw_{2}} \circ_{h} assoc$}    (e);
 	\end{tikzpicture}
 	
 	
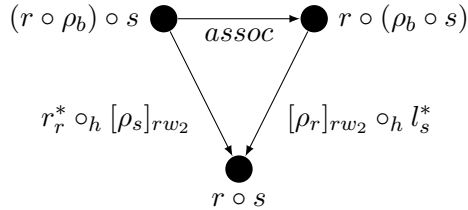
\bigskip
 	
 	\begin{tikzpicture}[>=latex,
 	every node/.style={auto},
 	arrowstyle/.style={double,->,shorten <=3pt,shorten >=3pt},
 	mydot/.style={circle,fill}]
 	
 	\coordinate[mydot,label=left:$(r \circ \rho_{b}) \circ s$](a)  at (0,0);
 	\coordinate[mydot,label=right:$r \circ (\rho_{b} \circ s)$](b) at (2,0);
 	\coordinate[mydot,label=below:$r \circ s$](c) at (1,-2);
 	
 	\draw[->]   (a) to coordinate (s) node[,swap] {$assoc$}    (b);
 	\draw[->]   (a) to coordinate (t) node[,swap] {$r^{*}_{r} \circ_{h} [\rho_{s}]_{rw_{2}}$}    (c);
 	\draw[->]   (b) to coordinate (x) node[] {$[\rho_{r}]_{rw_{2}} \circ_{h} l^{*}_{s}$}    (c);
 	\end{tikzpicture}
 \end{center}
 
 The proofs are straightforward. Remember that $assoc$ is just an application of $\sigma(tt)$, $r^{*}_{r}$ is an application of $trr$ and $l^{*}_{s}$ an application of $tlr$. First, let's start with $((u\circ p) \circ r) \circ s = \tau(s,\tau(r,\tau(p,u)))$ going to the right of the diagram:
 \begin{center}
 	$(assoc \circ_{h} [\rho_{s}]_{rw_{2}})(\tau(s,\tau(r,\tau(p,u)))) = \tau(s,assoc(\tau(r,\tau(p,u)))) =$
 	
 	$\tau(s,\tau(\tau(r,p),u))$
 	
 	$assoc(\tau(s,\tau(\tau(r,p),u))) = \tau(\tau(s,\tau(r,p)),u)$
 	
 	$([\rho_{u}]_{rw_{2}} \circ_{h} assoc)(\tau(\tau(s,\tau(r,p)),u)) = \tau(assoc(\tau(s,\tau(r,p))), u) = $
 	
 	$\tau(\tau(\tau(s,r),p)), u) = u \circ (p \circ (r \circ s))$
 \end{center}
 
 Now, starting from the same $ \tau(s,\tau(r,\tau(p,u)))$ and going bottom left:
 
 \begin{center}
 	$assoc(\tau(s,\tau(r,\tau(p,u)))) = \tau(\tau(s,r),\tau(p,u))$
 	
 	$assoc(\tau(\tau(s,r),\tau(p,u))) = \tau(\tau(\tau(s,r),p),u) = u \circ (p \circ (r \circ s))$
 \end{center}
 
 Therefore, the first diagram commutes. We now need to check the second one. Let's start from $((r \circ \rho_{b}) \circ s) = \tau(s,\tau(\rho_{b},r))$ and going to the right of the diagram:
 
 \begin{center}
 	$assoc(\tau(s,\tau(\rho_{b},r))) = \tau(\tau(s,\rho_{b}),r)$
 	
 	$([\rho_{r}]_{rw_{2}} \circ_{h} l^{*}_{s})\tau(\tau(s,\rho_{b}),r) = \tau(l^{*}_{s}(\tau(s,\rho_{b})),r) = $
 	
 	$\tau(s,r) = r \circ s$
 \end{center}
 
 Now, starting from the same $\tau(s,\tau(\rho_{b},r))$ and going right bottom:
 
 \begin{center}
 	$(r^{*}_{r} \circ_{h} [\rho_{s}]_{rw_{2}})\tau(s,\tau(\rho_{b},r)) = \tau(s,r^{*}_{r}(\tau(\rho_{b},r))) =$
 	
 	$\tau(s,r) = r \circ s$
 \end{center}
 
 Thus, the second diagram commutes. The coherence laws hold.

\end{appendices}

\end{document}